\documentclass[letterpaper, 10 pt, conference]{ieeeconf}  

\IEEEoverridecommandlockouts                              

\usepackage{amsmath,amssymb,amsfonts}
\usepackage{algorithmic}
\usepackage{graphicx}
\usepackage{textcomp}
\def\BibTeX{{\rm B\kern-.05em{\sc i\kern-.025em b}\kern-.08em
    T\kern-.1667em\lower.7ex\hbox{E}\kern-.125emX}}
    
\usepackage{graphicx,epstopdf}

\usepackage{amssymb,amsthm}
\usepackage{amsmath}
\usepackage{amsfonts}
\usepackage{mathrsfs}
\usepackage{tabularx} 
\usepackage[linesnumbered,boxed,ruled,commentsnumbered]{algorithm2e}

\usepackage{subfigure}
\usepackage{acronym}
\usepackage{todonotes}
\usepackage{color}
\usepackage{mathtools}
\usepackage{booktabs}
\usepackage{multirow}

\newtheorem{theorem}{Theorem}

\newtheorem{prop}{Proposition}
\newtheorem{defi}{Definition}
\newif\ifuseboldmathops
\newif\ifuseittextabbrevs

\ifuseittextabbrevs

\else

\fi

\ifuseboldmathops

\else

\fi

\ifuseboldmathops

\else

\fi

\ifuseboldmathops
	\newcommand{\Expect}{\mathop{\bf E{}}\nolimits}
	
\else
	\newcommand{\Expect}{\mathop{\mathbb{E}{}}\nolimits}
	
\fi

\ifuseboldmathops

\else

\fi

\newcommand{\calF}{\mathcal{F}}

\acrodef{mdp}[MDP]{Markov decision process}
\acrodef{dfa}[DFA]{deterministic finite-state automaton}
\acrodef{ltl}[LTL]{linear temporal logic}
\acrodef{ltlf}[LTL$(\calF)$]{quantitative linear temporal logic}
\acrodef{ag}[AG]{Assume-Guarantee}
\acrodef{ssp}[SSP]{Stochastic Shortest Path}
\acrodef{mcmc}[mcmc]{Monte Carlo Markov chain}

\theoremstyle{definition}
 
 \newtheorem{example}{Example}

\acrodef{ltl}[LTL]{linear temporal logic formula}
\acrodef{mdp}[MDP]{Markov decision process}
\acrodef{smdp}[Semi-MDP]{Semi-Markov decision process}

\acrodef{scltl}[scLTL]{syntactically co-safe LTL}
\newcommand{\defeq}{\mathrel{\mathop:}=}

\title{\LARGE \bf
Stable and Efficient Shapley Value-Based Reward Reallocation for Multi-Agent Reinforcement Learning of Autonomous Vehicles
}

\author{Songyang Han \and He Wang  \and Sanbao Su \and Yuanyuan Shi \and Fei Miao
\thanks{This work was supported by NSF 1849246, NSF 1952096, and NSF 2047354 grants. S.~Han, S.~Su, F.~Miao are with the Department of Computer Science and Engineering, University of Connecticut, Storrs Mansfield, CT, USA.  Email: \{songyang.han, sanbao.su, fei.miao\}@uconn.edu. H.~Wang is with the School of Information Science and Technology, ShanghaiTech University, Shanghai, China. Email: \{wanghe\}@shanghaitech.edu.cn. Y.~Shi is with the Electrical and Computer Engineering Department, University of California, San Diego, La Jolla, CA, USA. Email: \{yyshi\}@eng.ucsd.edu.}}

\begin{document}

\maketitle
\thispagestyle{empty} 
\pagestyle{empty}

\begin{abstract}
With the development of sensing and communication technologies in networked cyber-physical systems (CPSs), multi-agent reinforcement learning (MARL)-based methodologies are integrated into the control process of physical systems and demonstrate prominent performance in a wide array of CPS domains, such as connected autonomous vehicles (CAVs). However, it remains challenging to mathematically characterize the improvement of the performance of CAVs with communication and cooperation capability. When each individual autonomous vehicle is originally self-interest, we can not assume that all agents would cooperate naturally during the training process. In this work, we propose to reallocate the system's total reward efficiently to motivate stable cooperation among autonomous vehicles. We formally define and quantify how to reallocate the system's total reward to each agent under the proposed transferable utility game, such that communication-based cooperation among multi-agents increases the system's total reward. We prove that Shapley value-based reward reallocation of MARL locates in the core if the transferable utility game is a convex game. Hence, the cooperation is stable and efficient and the agents should stay in the coalition or the cooperating group. We then propose a cooperative policy learning algorithm with Shapley value reward reallocation. In experiments, compared with several literature algorithms, we show the improvement of the mean episode system reward of CAV systems using our proposed algorithm.
\end{abstract}

\section{introduction}
\label{sec:intro}
The rapid evolution of ubiquitous sensing, communication, and computation technologies has contributed to the revolution of cyber-physical systems (CPSs). Increasingly, multi-agent reinforcement learning (MARL)-based methodologies are integrated into the control process of physical systems and demonstrate prominent performance in a wide array of CPS domains. Connected and autonomous vehicles (CAVs) are one type of networked CPSs and multi-agent systems, with the development of vehicle-to-everything (V2X) communication technologies. The U.S. Department of Transportation (DOT) has estimated that DSRC (dedicated short-range communications) based V2V communication can address up to 82\% of all crashes in the U.S. and potentially save thousands of lives and billions of dollars~\cite{DSRC_standard}. Information sharing of basic safety messages (BSMs) benefits CAVs coordination and control approaches in scenarios such as cross intersections or lane-merging~\cite{Coordinate_CAV, CV_intersection, ort2018autonomous}, platoon and adaptive cruise control (ACC)~\cite{cacc_platoon, savefuel_Karl, V2Vbenefit_platoon}. 

However, it remains challenging to formally characterize the improvement of learning-based decision-making for CAV systems with V2X connections and to reallocate the system's total reward efficiently to motivate stable cooperation of individual autonomous vehicles. Existing learning-based planning or control approaches do not utilize communication or potential shared information for autonomous vehicles yet~\cite{waymo_rss19,RobustEndtoEnd_icra20,PredictPolicyLearn_LeCun,UCBFlow_iccps19}. How V2X communication benefits MARL, improves the system's total reward, and motivates cooperation has not been formally defined or quantified.

We propose to formally define and quantify the value of communication-based cooperation to MARL based on Shapley value~\cite{Shapley, ShapleyQ_aaai20}, and use it to reallocate the system's total reward to the individual agent to motivate the cooperation. It will build a research foundation to formally quantify the value of information sharing and motivate cooperation for the CAV research society, and the results can be leveraged to other networked CPSs to better understand the benefits of cooperation. Our proposed approach includes the following three major novelties and contributions.
\begin{itemize}
    \item We formally define and quantify how to efficiently reallocate the system's total reward to the individual agent to motivate information sharing and stable cooperation among multi-agents.
    \item We define a transferable utility game formulation to study the reward reallocation problem, and prove that Shapley value-based reward reallocation is efficient in Theorem~\ref{theorem:efficient}. We prove that the reward reallocation scheme is stable if the transferable utility game is a convex game in Theorem~\ref{theorem:core}. Hence, each individual agent should cooperate to receive more rewards.
    \item We propose a cooperative policy learning with Shapley value reward allocation algorithm. In experiments, we show the improvement of the system's total reward, velocity, and comfort of CAV systems using the Shapley value of the grand coalition (all agents cooperate).
\end{itemize}

\section{Related Work}
\label{sec: relate}
Cooperative games attract increasing research interests in MARL~\cite{zhang2019multi}. Existing works have been investigated to encourage every agent to work collaboratively by assigning rewards appropriately, such as a value decomposition network~\cite{sunehag2018value, rashid_nips20}, subtracting a counterfactual baseline~\cite{foerster2017counterfactual}, or an implicit method~\cite{NEURIPS2020_8977ecbb}. Multi-agent deep deterministic policy gradient (MADDPG) applies a centralized $Q$-function to address the problem caused by the non-stationary environment~\cite{lowe2017multi}. Its scalability can be improved by adding an attention mechanism to the centralized critic~\cite{iqbal2019actor}. Distributed execution with communication among agents with centralized critics methods and online learning with communicative actions~\cite{learnComm_nips16, learnComm_iclr19,communication_POMDP, Comm_MAplan, learnComm_MARL, attentionComm_nips18} have shown performance improvement over non-communicative agents. In all these works, agents are assumed to cooperate during the training process, however, whether agents will cooperate or whether the cooperative coalition is \textbf{stable} has not been answered yet. 

Another related line of work is coalition formation games, where selfish and rational agents choose to participate in coalitions to maximize their own utility. There are two classes of coalition formation games: static and dynamic formation processes. 
In the former class, several works employ bargaining methods \cite{okada96,madoc93} to achieve the agreement 
and no agents would deviate from the agreement once it begins. Besides, some works rely on concepts of cooperative game theory to acquire stable coalitions and fair payoff distribution methods, such as Shapley value \cite{vakilinia2019fair}, core \cite{sandholm1999distributed,vakilinia2019fair} and kernel \cite{taase04}. 
However, these methods do not apply to our setting because we consider a dynamic process.
Methods that be used for analyzing dynamic coalition formation include Bayesian RL and MARL~\cite{chalkiadakis04,Taywade21}. But these works all use MARL as a tool to analyze the coalition formation game, while we consider the converse - we leverage analytical tools in cooperative game and coalition formation game for reward reallocation design in MARL.

Shapley value~\cite{Shapley} is one classical method to divide the total payoff of all players in a cooperative game so that each player receives his or her fair payment. It uniquely provides an equitable assignment of values to agents and is also one metric to measure the importance and marginal contribution of each player to the system. Recent literature on learning has developed a principled framework of ``data Shapley"~\cite{DataShapley_ml19} to address data valuation and ``feature Shapley"~\cite{ShapleyFeature_nips17} to measure the importance of features to supervised learning algorithms. 
Computing Shapley is still challenging and requires computing all the possible marginal contributions, which is exponentially large in the training process. Therefore, approximation methods to estimate the Shapley value of agents, data, and features have been designed and are presented in the game theory and machine learning literature~\cite{ShapleyQ_aaai20,computeShapley_13, DataShapley_ml19, ShapleyFeature_nips17}. However, in this work, we leverage the exploration and exploitation power of reinforcement learning to reduce the computational complexity significantly by avoiding traversing all possible marginal contributions.


\section{Preliminary}
\label{sec:preliminary}

\subsection{Convex Game}
\label{sec:convex_game}
In cooperative games, a Transferable Utility game (\textbf{TU game}) with $n$ agents can be represented by a pair $(\mathcal{N},v)$, where $\mathcal{N} = \{1,\ldots,n\}$ is a  set of agents and $v:2^{n}\to \mathbb{R}$ is the real-valued characteristic function~\cite{chalkiadakis2011computational}. More specifically, $\mathcal{N}$ is the grand coalition in which all the $n$ agents cooperate together. For any coalition $\mathcal{C}\subseteq \mathcal{N}$, $v(\mathcal{C})$ measures its value. A TU game $(N,v)$ is a \textbf{convex game} if for any $\mathcal{C,D}\subseteq \mathcal{N}$, $v(\mathcal{C}\cup \mathcal{D})\ge v(\mathcal{C})+v(\mathcal{D})-v(\mathcal{C}\cap \mathcal{D})$. An \textbf{outcome} of any TU game consists of two parts, a \textbf{coalition structure} $\mathcal{CS}$ and a \textbf{payoff vector} $x\in\mathbb{R}^n$, denoted as a pair $(\mathcal{CS},x)$. The coalition structure $\mathcal{CS} = \{\mathcal{C}^1,\ldots,\mathcal{C}^k\}$ over $\mathcal{N}$ is a possible partition of $\mathcal{N}$ and the corresponding payoff vector $x = (x^1,\ldots,x^n)\in\mathbb{R}^n$ satisfies: $(1)$ $x^i\ge 0 $ for each $i\in\mathcal{N}$ and $(2)$ $x(\mathcal{C}^j)\le v(\mathcal{C}^j)$, where $x(\mathcal{C}^j) = \sum_{i\in\mathcal{C}^j} x^i$, for any $j\in\{1,\ldots,k\}$. 

\subsection{Solution Concepts}
There are several solution concepts that identify sets of outcomes with appealing properties, such as stability and efficiency. In this paper, we investigate four widely used solution concepts~\cite{chalkiadakis2011computational}, as follows.
\begin{defi}[\textbf{Core and Stable Outcome}]
\label{def:stable}
Given a TU game $(\mathcal{N},v)$, the Core is the set of all stable outcomes $(\mathcal{CS},x)$ such $x(\mathcal{C})\ge v(\mathcal{C})$ for every $\mathcal{C}\subseteq \mathcal{N}$, where $x(\mathcal{C}) = \sum_{i\in\mathcal{C}} x^i$.
\end{defi}
\begin{defi}[\textbf{Efficient Outcome}]
\label{def:efficient}
If $x(\mathcal{C}) =  v(\mathcal{C})$ for any $\mathcal{C}\in \mathcal{CS}$ in an outcome $(\mathcal{CS},x)$, then this outcome is efficient and maximizes the social welfare.
\end{defi}
\begin{defi}[\textbf{Shapley Value}]
\label{def:Shaply_value_old}
Given a TU game $(\mathcal{N},v)$, the Shapley Value for each agent $i\in\mathcal{N}$ is denoted by $\phi^i(v)$ and is given by
\begin{equation*}
    \phi^i(v) = \sum_{\mathcal{C}\subseteq \mathcal{N}\setminus\{i\}} \frac{|\mathcal{C}|!(|\mathcal{N}|-|\mathcal{C}|-1)!}{|\mathcal{N}|!}\left(v(\mathcal{C}\cup\{i\}) - v(\mathcal{C})\right).
\end{equation*}
\end{defi}


\subsection{Multi-agent Actor-Critic}
The environment's state transition of MARL is influenced by the policy of all agents and it is non-stationary from a single agent's view. To alleviate this problem and stabilize training, MADDPG is proposed using a centralized $Q$-function that has global state information~\cite{lowe2017multi}. The $i$th agent maximizes its own expected return $J(\theta ^i)$ and its gradient is 
\begin{align}
\label{equ:actor_maddpg}
    & \nabla_{\theta^i}J(\theta^i) = \\
    &\Expect_{s, a \sim \mathcal{D}} \left[ \nabla_{\theta^i} \pi^i(s^i)\nabla_{a^i}Q^i(s,a^1,...,a^n)|_{a^i = \pi^i(s^i)} \right], \nonumber
\end{align}
where $Q^i(s,a^1,...,a^n)$ is a centralized action-value function, $s$ is the joint state, $\pi^i$ is the policy. The experience replay buffer $\mathcal{D}$ contains $(s, a^1, ..., a^n, s', r^1, ..., r^n)$. The centralized critic $Q^i$ is trained using the Bellman loss:
\begin{align}
\label{equ:critic_maddpg}
    & \mathcal{L}(\theta^i) = \Expect_{s, a, r, s' \sim \mathcal{D}} [y - Q^i(s, a^1,...,a^n)]^2, \nonumber \\
    & y = r^i + \gamma Q^{i\prime}(s^{\prime}, a^{1\prime},..., a^{n\prime})|_{a^{j\prime} = \pi^{j}(s^{j\prime})},
\end{align}
where $Q^{i\prime}$ is the target network and $\gamma$ is a discount factor. Note that this algorithm adopts the centralized training and decentralized execution paradigm. When testing, each agent can only access its local state to select actions.

\section{Problem Formulation}
\label{sec:formulation}
We first formulate the connected autonomous vehicles (CAVs) problem as a MARL problem. We consider $n$ agents (e.g., autonomous vehicles) in the agent set $\mathcal{N} = \{1, ..., n\}$. Each agent $i$ is associated with an action $a^i \in \mathcal{A}^i$ and a state $s^i \in \mathcal{S}^i$. The global joint state is $ s = (s^1, ..., s^n) \in \mathcal{S} \defeq \mathcal{S}^1 \times \cdots \times \mathcal{S}^n$. Similarly, the global joint action is $ a = (a^1, ..., a^n) \in \mathcal{A} \defeq \mathcal{A}^1 \times \cdots \times \mathcal{A}^n$. Each coalition has a stage-wise reward function $R(s, a, \mathcal{C})$, where we denote a subset of the vehicles $\mathcal{N}$ as a coalition $\mathcal{C}$. 
Specially, the $\mathcal{N}$ denotes the grand coalition in which all the $n$ agents cooperate together. Let $R(s,a, \mathcal{N})$ to represent the global stage-wise reward under the grand coalition. For the empty coalition $\emptyset$, $R(s,a, \emptyset) \defeq 0$. In our work, we assume that all rewards are nonnegative. 

We consider each agent is associated with a localized policy $\pi^i(a^i| s^\mathcal{C})$ to choose a local action $a^i$ given its local state $s^\mathcal{C}$. We use $\pi(a | s)$ to denote the joint policy. The system's objective is to find a policy $\pi(a | s)$ to maximize the system's discounted total reward under the grand coalition, 
\begin{equation}
\label{equ:total_reward}
   \max_{\pi(a|s)} \Expect_{a_t \sim \pi(a|s)} \left[ \sum_{t=0}^\infty \gamma^t R_{t+1}(s_t, a_t, \mathcal{N}) | s_0 = s, a_0 = a \right], \nonumber
\end{equation}
where $\gamma$ is the discount factor. However, how to motivate agents to form a grand coalition has not been studied yet.


The MARL problem is usually solved by the centralized training and decentralized execution paradigm that is first proposed in~\cite{lowe2017multi}. During training, there is a centralized critic $Q(s,a)$ that has access to the global state $s$ and global action $a$. This critic is used to train a localized actor for decentralized execution. However, having such a powerful centralized critic for autonomous vehicles is not easy, especially when each vehicle is self-interested and focuses to maximize its own reward (e.g., maximize its own driving speed or minimize its own travel time). Vehicles are not assumed to be fully cooperative.
Therefore, the goal of this paper is to design a better total reward reallocation mechanism, under which agents are willing to collaborate with each other by sharing their state and action information.

Simply using the original action value function $Q(s,a)$ is not guaranteed to learn a policy where agents are willing to cooperate. In order to find a better total reward for policy learning, we formulate and analyze the game between CAVs using the cooperative game framework. We consider a Transferable Utility game (\textbf{TU game}) $G = (\mathcal{N}, v)$, where $v$ is the characteristic function introduced in Section~\ref{sec:convex_game}. 
Consider the entire MARL problem across all timesteps, the value of any coalition $\mathcal{C}$ given the current state and action is 
\begin{equation}
    v(\mathcal{C}|s, a) = \Expect_{\pi(a|s)} \left[ \sum_{t=0}^\infty \gamma^t R_{t+1}(s_t, a_t, \mathcal{C}) | s_0 = s, a_0 = a\right]. \nonumber
\end{equation}

Specially, when $\mathcal{C} = \mathcal{N}$ we have 
\begin{equation}
    v(\mathcal{N}|s, a) = \Expect_{\pi(a|s)} \left[ \sum_{t=0}^\infty \gamma^t R_{t+1}(s_t, a_t, \mathcal{N}) | s_0 = s, a_0 = a\right]. \nonumber
\end{equation}
It shows the value of the grand coalition $\mathcal{N}$ is exactly the system's total reward. Note that the characteristic function is not the state value function V(s) used in reinforcement learning. The characteristic function is also a function of the coalition $\mathcal{C}$ in the TU game.

We want to find a stable and efficient total reward reallocation for this TU game such that agents can learn a cooperative policy. In the following section, we show how to utilize the Shapley value to solve this problem.






\section{Stable and Efficient Reward Reallocation}
\label{sec:reward_reallocation}




We first define the Shapley value for the Transferable Utility game (TU game) $G=(\mathcal{N}, v)$ defined in Section~\ref{sec:formulation} and show that the Shapley value satisfies the axioms for a fair reward reallocation. We then present the Shapley value is an efficient reward reallocation in Theorem~\ref{theorem:efficient}. Moreover, we show the Shapley value is a stable reward reallocation if the TU game $G=(\mathcal{N}, v)$ is a convex game in Theorem~\ref{theorem:core}. Then we give an example of the convex game and design Algorithm~\ref{alg:actor-critic}. An efficient solution means the system's total reward is completely distributed to each agent for cooperative policy learning. A stable solution means any agent cannot get more payoff if they leave the coalition. 
For CAVs, a stable solution means that vehicles will stay within the coalition, communicate and cooperate with other coalition members to optimize the coalition-level objective.



We define the Shapley value of the TU game $G= (\mathcal{N}, v)$ as follows.
\begin{defi}[\textbf{Shapley Value}]
\label{def:Shapley_value_new}
The Shapley value of the TU game $G= (\mathcal{N}, v)$ is defined to be
\begin{equation}
\label{equ:Shapley}
    \phi^i(s, a)\! \defeq \!\sum_{\mathcal{C}\subseteq \mathcal{N}\setminus\{i\}} \!\!\!\!\!\frac{|\mathcal{C}|!(n\!-\!|\mathcal{C}|\!-\!1)!}{n!} [v(\mathcal{C}\!\cup\! \{i\}|s, a) \!-\! v(\mathcal{C}|s,a)].
\end{equation}
\end{defi}

The Shapley value of the TU game $G= (\mathcal{N}, v)$ is extended for the multi-agent sequential decision problem from Definition~\ref{def:Shaply_value_old}. The intuition is that each agent's total reward is proportional to its total contribution to the entire system. It is used to quantify the contribution of each agent for the communication-based cooperation in MARL. Based on the definition of $\phi^i(s,a)$ in~\eqref{equ:Shapley}, it is straightforward to see that the Shapley value satisfies the axioms for a fair reward reallocation as follows:

\begin{prop}
\label{prop:fairness}
The Shapley value of the TU game $G = (\mathcal{N},v)$ satisfies the following axiomatic characterization for a fair reward reallocation:
\begin{enumerate}
    \item Symmetric: if $v(\mathcal{C} \cup \{i\}|s, a) = v(\mathcal{C} \cup \{j\}| s, a)$ for any coalition $\mathcal{C} \subseteq \mathcal{N} \setminus\{i, j\}$, then $\phi^i(s, a) = \phi^j(s, a)$.
    \item Dummy player: if $v(\mathcal{C} \cup \{i\}|s, a) = v(\mathcal{C} | s, a)$ for any coalition $\mathcal{C} \subseteq \mathcal{N}$, then $\phi^i(s, a) = 0$. 
    \item Additivity: for any two $v^1$ and $v^2$, $\phi^i(s, a |v^1 + v^2) = \phi^i(s, a |v^1 ) + \phi^i(s, a | v^2)$ for each $i$, where $\phi^i(s, a |v^1 + v^2)$ is Shapley value of the TU game $(\mathcal{N}, v^1 + v^2)$ and $(v^1 + v^2)( \mathcal{C}|s, a) = v^1( \mathcal{C}|s, a) + v^2( \mathcal{C}|s, a)$.
\end{enumerate}
\end{prop}

If a reward reallocation satisfies the axioms of symmetric, dummy player, and additivity, it is called a fair reward reallocation~\cite{chalkiadakis2011computational}. Besides fairness, we show that the \textbf{outcome} $( \{\mathcal{N} \}, \phi)$ with Shapley value is an efficient outcome of the TU game $G= (\mathcal{N}, v)$, where $\{\mathcal{N} \}$ is the coalition structure for the grand coalition, and $\phi(s,a) = (\phi^1(s,a), ..., \phi^n(s,a))$ is the payoff vector for each agent given the current state $s$ and action $a$. The payoff $\phi^i(s,a)$ for each agent is the reallocated total reward. It is regarded as a new action value function after the reward reallocation.
\begin{theorem} 
\label{theorem:efficient}
The outcome $( \{\mathcal{N} \}, \phi)$ with Shapley value is an efficient outcome of the TU game $G= (\mathcal{N}, v)$ and the Shapley value is an efficient reward reallocation.
\end{theorem}
\begin{proof}
Note that $v(\emptyset|s,a) = 0$, we have
\begin{align}
      \sum_{i \in \mathcal{N}} \phi^i(s, a) &= \sum_{i \in \mathcal{N}}  \sum_{\mathcal{C}\subseteq \mathcal{N}\setminus\{i\}} \!\!\!\!\!\frac{|\mathcal{C}|!(n\!-\!|\mathcal{C}|\!-\!1)!}{n!} [v(\mathcal{C}\!\cup\! \{i\}|s, a) \!\nonumber\\
      -\! v(\mathcal{C}|s,a)] &= n \frac{(n-1)!1!}{n!} v(\mathcal{N}|s,a) + n \frac{0!(n-1)!}{n!} v(\emptyset|s,a) \nonumber\\
    & + \sum_{\mathcal{C}\subset \mathcal{N}, \mathcal{C} \neq \emptyset, |\mathcal{C}|=p} ( p \frac{(p-1)!(n-p)!}{n!} -  \nonumber \\
    & (n-p) \frac{p!(n-p-1)!}{n!}) v(\mathcal{C}|s, a) \nonumber \\
    & = v(\mathcal{N}|s,a).
\end{align}
Note that the first equation is a telescoping sum where both $v(\mathcal{N}|s,a)$ and $v(\emptyset|s,a)$ appear once for each agent. The value of coalition $\mathcal{C}$ with $p$ agents appears $p$ times with positive sign, once for each agent in $\mathcal{C}$; it also appears $n-p$ times with negative sign, once for each agent not in $\mathcal{C}$.
\end{proof}

Moreover, if the TU game $G= (\mathcal{N}, v)$ is a convex game, we show that the Shapley value is a stable reward reallocation. A convex game is defined as follows:
\begin{defi}[\textbf{Convex Game}]
The TU game $G =(\mathcal{N},v)$ is a convex game if for any pair of coalitions $\mathcal{C},\mathcal{D}\subseteq \mathcal{N}$,
\begin{equation*}
    v(\mathcal{C}\cup\mathcal{D}|s, a) + v(\mathcal{C}\cap\mathcal{D}|s, a) \ge v(\mathcal{C}|s, a) + v(\mathcal{D}|s, a).
\end{equation*}
\end{defi}

For a convex game, we have the following theorem that guarantees the core is nonempty, and there exists a stable solution for reallocating $v(\mathcal{N}|s,a)$ to each individual agent to motivate cooperation. 

\begin{theorem} 
\label{theorem:core}
If the TU game $G= (\mathcal{N}, v)$ is a convex game, the outcome $( \{\mathcal{N} \}, \phi)$ with Shapley value is in the core and the Shapley value is a stable reward reallocation.
\end{theorem}
\begin{proof}
Note that the Shapley value for the TU game $G = (\mathcal{N}, v)$ is 
\begin{equation}
    \phi^i(s, a)\! = \!\sum_{\mathcal{C}\subseteq \mathcal{N}\setminus\{i\}} \!\!\!\!\!\frac{|\mathcal{C}|!(n\!-\!|\mathcal{C}|\!-\!1)!}{n!} [v(\mathcal{C}\!\cup\! \{i\}|s, a) \!-\! v(\mathcal{C}|s,a)]. \nonumber
\end{equation}

We then show that Shapley value allocation lies in the core of the TU game $G= (\mathcal{N}, v)$. 

Let the payoff vector defined by the marginal contribution of each agent, i.e., $x^i = v(\mathcal{C}\cup \{i\}|s,a) - v(\mathcal{C}|s,a)$ for $i\in\mathcal{N}$. The goal is to prove that $x=(x^1,\ldots,x^n)$ is in the core. For any  coalition $\mathcal{C}=\{i^1,i^2,\ldots,i^k\} \subseteq \mathcal{N}= \{1, ..., n\}$, we can write $v(\mathcal{C})$ as follows by telescoping sum, 
\begin{align}
    & v(\mathcal{C}|s,a) = v(\{i^1\}|s,a)-v(\emptyset|s,a) + v(\{i^1,i^2\}|s,a) \nonumber \\
    &- v(\{i^1\}|s,a) +\cdots+v(\mathcal{C}|s,a) - v(\mathcal{C}\setminus\{i^k\}|s,a).
\end{align}

Without loss generality, we assume that $i^1<i^2<\cdots<i^k$. Then for each $j \in\{1,\ldots,k\}$, we let $\mathcal{D}=\{1, 2, \ldots, i^{j-1}\}$ and by the definition of convex games, 
\begin{align}
    &v(\{i^1,\ldots,i^{j-1},i^j\}|s,a) - v(\{i^1,\ldots,i^{j-1}\}|s,a) \nonumber\\
    \le\ & v(\{1,\ldots,i^{j-1},i^j\}|s,a ) - v(\{1,\ldots,i^{j-1}\}|s,a)\nonumber\\
    =\ &x^{i^j},\label{eq:thm1_shapley_core}
\end{align}
where the last equation is based on $x^{i^j} = v(\mathcal{D}\cup \{i^j\}|s,a) - v(\mathcal{D}|s,a)$. By summing up \eqref{eq:thm1_shapley_core} from $j=1$ to $k$, we have $v(\mathcal{C}|s,a)\le \sum_{j=1}^k x^{i^j} = x(\mathcal{C}|s,a)$. Therefore, $x$ is a stable solution in the core based on Definition~\ref{def:stable} and any coalition $\mathcal{C}$ will not have incentives to deviate from the grand coalition. Moreover, since Shapley value is a convex combination of our constructed $x$, and core is a convex set, the outcome $( \{\mathcal{N} \}, \phi)$ with Shapley value is also stable and lies in the core. 


\end{proof}

According to Theorem~\ref{theorem:efficient} and Theorem~\ref{theorem:core}, the outcome $(\{\mathcal{N} \}, \phi)$ with Shapley value is a stable and efficient outcome of the TU game $G= (\mathcal{N}, v)$ if $G$ is a convex game. No agent wants to deviate from the grand coalition in this case. Then, using the action value function $Q(s,a)$ for the policy learning is not the best choice. We can reallocate the system's total reward to each agent using the Shapley value $\phi^i(s, a)$. It is a fair, stable, and efficient total reward reallocation for a convex game. Now, we give an example for the convex game. 
\begin{prop}
\label{prop:convex_example}
When $R(s, a, \mathcal{C})= \sum_{i\in \mathcal{C}} R(s, a, \{i\})$ for any nonempty coalition $\mathcal{C} \subseteq \mathcal{N}$, the TU game $G =(\mathcal{N},v)$ is a convex game.
\end{prop}
\begin{proof}
For any nonempty coalition $\mathcal{C} \subseteq \mathcal{N}$, $s\in\mathcal{S}$, and $a\in\mathcal{A}$, we have
\begin{align}
\label{equ:characteristic}
    &v(\mathcal{C}|s, a) = \Expect_{a_t \sim \pi(a|s)} \left[ \sum_{t=0}^\infty \gamma^t R_{t+1}(s_t, a_t, \mathcal{C}) | s_0 = s, a_0 = a\right] \nonumber \\
    &=  \sum_{i\in{\mathcal{C}}} \Expect_{a_t \sim \pi(a|s)} \left[ \sum_{t=0}^\infty \gamma^t R_{t+1}(s_t, a_t, \{i\}) | s_0 = s, a_0 = a \right] \nonumber \\
    & = \sum_{i\in{\mathcal{C}}} v(\{i\}|s, a).
\end{align}
Then we have for any pair of nonempty coalitions $\mathcal{C},\mathcal{D}\subseteq \mathcal{N}$, $s\in\mathcal{S}$, and $a\in\mathcal{A}$
\begin{align}
    & v(\mathcal{C}\cup\mathcal{D}|s, a) + v(\mathcal{C}\cap\mathcal{D}|s, a) \nonumber \\
     =& \sum_{i\in\mathcal{C}\cup\mathcal{D}} v(\{i\}|s, a) + \sum_{i\in\mathcal{C}\cap\mathcal{D}} v(\{i\}|s, a)\nonumber \\
    =& \left( \sum_{i\in\mathcal{C}} v(\{i\}|s, a) +  \sum_{i\in\mathcal{D}\setminus \mathcal{C}} v(\{i\}|s, a) \right) +  \sum_{i\in \mathcal{C}\cap\mathcal{D}} v(\{i\}|s, a) \nonumber \\
    =& \sum_{i\in\mathcal{C}} v(\{i\}|s, a) +  \sum_{i\in\mathcal{D}} v(\{i\}|s, a) \nonumber \\
    \geq& v(\mathcal{C}|s, a) + v(\mathcal{D}|s, a). 
\end{align}
If the pair of coalitions $\mathcal{C},\mathcal{D}\subseteq \mathcal{N}$ has at least one empty set, it is easy to check $v(\mathcal{C}\cup\mathcal{D}|s, a) + v(\mathcal{C}\cap\mathcal{D}|s, a) \ge v(\mathcal{C}|s, a) + v(\mathcal{D}|s, a)$ considering $v(\emptyset|s,a) = 0$.
\end{proof}

Proposition~\ref{prop:convex_example} gives us one example of the convex game. There exits other kinds of convex games, for example, when $R(s, a, \mathcal{C})= - \max_{i\in \mathcal{C}} R(s, a, \{i\})$~\cite{littlechild1973simple} for any $\mathcal{C} \neq \emptyset$.

\begin{algorithm}[h]
\small
\SetAlgoLined
 Randomly initialize the characteristic network $v$ and the actor network $\pi^i$ for agent $i$. Initialize target networks $v^{\prime}, \pi^{i\prime}$\;
 \For {each episode}
 {
     Initialize a random process $\mathcal{X}$ for action exploration\;
     Receive initial state $\textbf{s}$\;
    \For {each timestep}
    { 
        Randomly sample a coalition $\mathcal{C}$. All agents in coalition $\mathcal{C}$ cooperate with each other at this time step.
        
        For each agent $i$, select action $a^i=\pi^i(s^\mathcal{C}) + \mathcal{X} $ w.r.t the current policy and exploration. Execute actions $\textbf{a} = (a^1,...,a^n)$ and observe the reward $\mathbf{r}$ and the new state information $\textbf{s}^{\prime}$. Store $(\textbf{s}, \mathbf{a}, \mathcal{C}, \mathbf{r}, \textbf{s}^{\prime})$ in replay buffer $\mathcal{D}$. Set $\textbf{s} \leftarrow \textbf{s}^{\prime}$\;
        \For {each agent}
        {
            Sample a random minibatch of samples $(\textbf{s}_k, \textbf{a}_k, \mathcal{C}_k, \mathbf{r}_k, \textbf{s}_k^{\prime})$ from $\mathcal{D}$\;
            
            
            
            Set $y_k = r^\mathcal{C}_k + \gamma v^{\prime}(\textbf{s}_k^{\prime}, \textbf{a}_k^{\prime}, \mathcal{C}_k)|_{a^{i\prime} = \pi^{i\prime}(s^{\mathcal{C}\prime})}$\;
            
            Update the characteristic network by minimizing the loss $\mathcal{L}(\theta) = \frac{1}{K}\sum_k (y_{k} - v(\textbf{s}_k, \textbf{a}_k, \mathcal{C}_k))^2$\;
            
            Calculate the Shapley value as $\phi^i(s, a)\! \defeq \!\sum_{\mathcal{C}\subseteq \mathcal{N}\setminus\{i\}} \!\!\!\!\!\frac{|\mathcal{C}|!(n\!-\!|\mathcal{C}|\!-\!1)!}{n!} [v(s, a, \mathcal{C}\!\cup\! \{i\}) \!-\! v(s,a, \mathcal{C})]$\;
            
            Update actor using the gradient $\nabla_{\theta^i}J \approx \frac{1}{K}\sum_k \nabla_{\theta^i} \pi^i(s^\mathcal{C})\nabla_{a^i}\phi^i(\textbf{s}_k, \textbf{a}_k)$ where $a^i = \pi^i(s^\mathcal{C})$\;
        }
        Update all target networks: $\theta^{i\prime} \leftarrow \tau\theta^i + (1-\tau)\theta^{i\prime}$.
    }
 }
 
 \caption{Cooperative Policy Learning with Shapley Value Reward Reallocation}
 \label{alg:actor-critic}
\end{algorithm}
\setlength{\textfloatsep}{1pt}

Based on the Shapley value defined in Definition~\ref{def:Shapley_value_new}, we propose the Algorithm~\ref{alg:actor-critic} for each vehicle to learn a cooperative policy. Instead of modeling the action value function, we use a neural network to learn the characteristic function $v(s, a, \mathcal{C})$. This is used to calculate the Shapley value in step 12. Then we use the Shaley value as the reallocated total reward to learn a cooperative policy, as we know the Shaley value reward reallocation is fair and efficient. Moreover, the Shaley value reward reallocation is stable for a convex game. This algorithm adopts centralized training and decentralized execution paradigm that is first proposed in~\cite{lowe2017multi}. The effectiveness of this algorithm is evaluated in the following experiment section.

\section{Experiment}
\label{sec:experiment}
In this section, we use CARLA~\cite{Dosovitskiy17}, an open-source simulator that supports the development, training, and validation of autonomous driving systems, to validate our proposed method. 
The host machine adopted in our experiments is a server configured with Intel Core i9-10900X processors and four NVIDIA Quadro RTX 6000 GPUs. Our experiments are performed on Python 3.5.4, GCC7 7.5, openAI gym 0.10.5, numpy 1.14.5, tensorflow 1.8.0, and CUDA 9.0. 

We consider a 3-lane freeway scenario with CAVs as shown in Fig.~\ref{fig:freeway_carla}. We use Algorithm~\ref{alg:actor-critic} for CAVs to learn a cooperative policy for behavior planning to decide whether to change or keep lane. For each vehicle, the action set $\mathcal{A}^i$ includes \{Keep Lane (KL), Change Left (CL), Change Right (CR), Emergency Stop (ES)\}. The stage-wise reward for each vehicle involves their velocity and comfort. The comfort of a vehicle (for passenger's experience) is defined based on its acceleration and action $a^i$ as follows:
\begin{equation}
\label{equ:comfort}
\small
    comfort = \begin{cases} 
			 3, & \text{if } \lvert acceleration \rvert< \Theta \text{ and } a^i = KL; \\
			2, & \text{if } \lvert acceleration \rvert \geq \Theta \text{ and } a^i = KL; \\
			 1, & \text{if } a^i = CL/CR; \\
			 0, & \text{if in } ES. \end{cases} 
\end{equation}
where $\Theta$ is a predefined threshold. The reward function for vehicle $i$ is defined as:
\begin{equation}
R(s, a, \{i\}) = w \cdot velocity + comfort,
\label{equ:reward}
\end{equation}
where $w$ is a trade-off weight. We randomly sample a coalition $\mathcal{C}$ at each timestep. All vehicles in the coalition $\mathcal{C}$ communicate and cooperate with other coalition members. Let $R(s, a, \mathcal{C})= \sum_{i\in \mathcal{C}} R(s, a, \{i\})$ to represent the coalition's state-wise reward. 
While training, we use Shapley value to reallocate the reward to encourage communication-based cooperation and improve the system's total reward in Eq.~\ref{equ:total_reward}. As analysed in Section~\ref{sec:reward_reallocation}, our reward reallocation method is fair (Proposition~\ref{prop:fairness}), efficient (Theorem~\ref{theorem:efficient}), and stable (Theorem~\ref{theorem:core}).

\begin{figure}[h]
  \centering
  \includegraphics[width=2in]{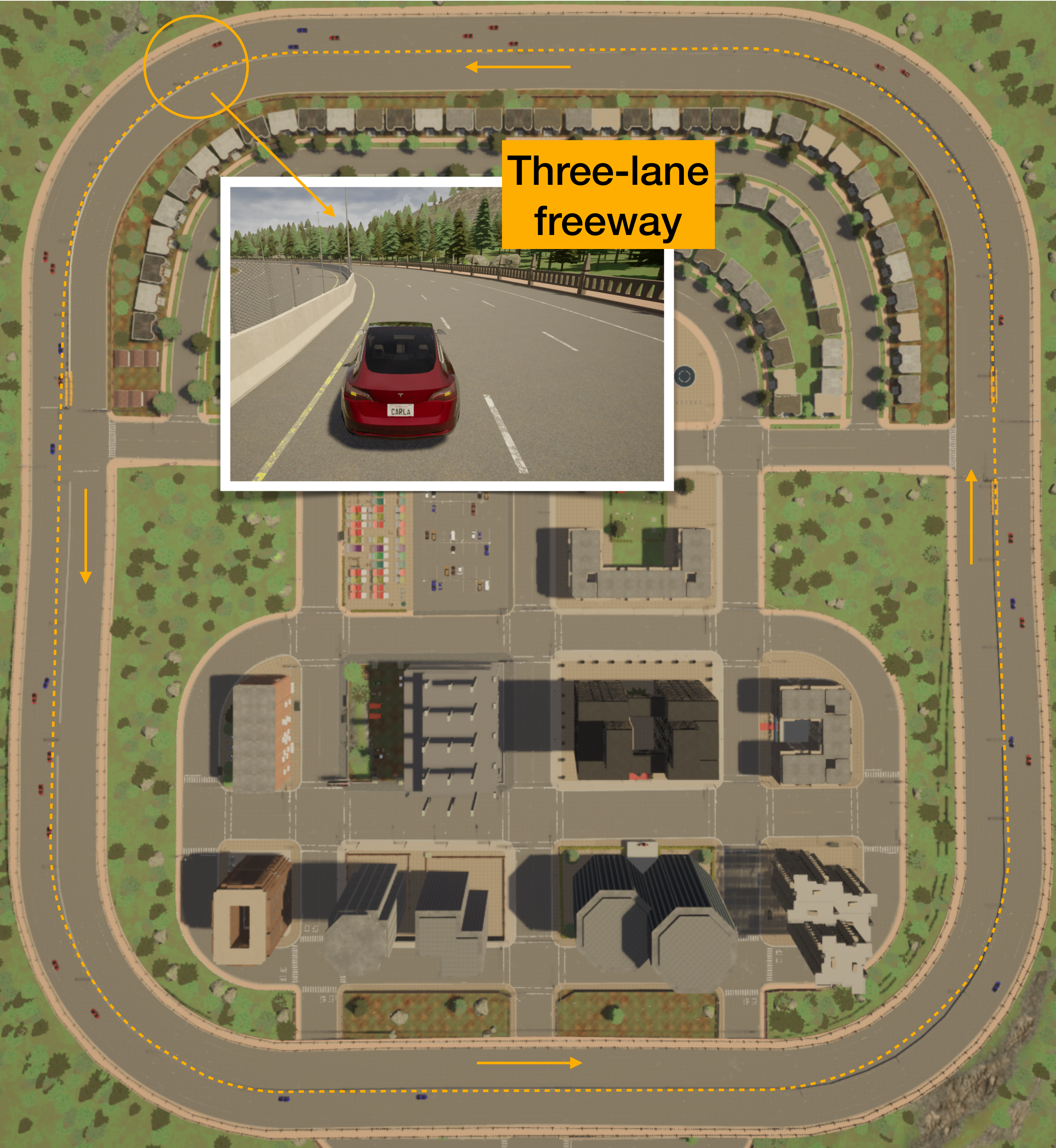}
  \caption{The example scenario of a 3-lane freeway in CARLA~\cite{Dosovitskiy17}. Vehicles are scattered on the outer-loop of the map ``Town 05". The environment can be mixed traffic with both autonomous and human-driven vehicles.}\label{fig:freeway_carla}
\end{figure}

Each vehicle's states include its position, velocity, acceleration, image and point clouds that are captured by the onboard camera and LIDAR sensors respectively. The resolution of camera image is $375\times1,242$ pixels. 
Each point cloud from LIDAR is stored with the 3 coordinates (the ego vehicle being the origin), representing forward, left, and up respectively, and an additional reflectance value. 
We use methods in~\cite{shanglin2021vision} to process images and point clouds together, with Ultra-Fast-Lane-Detection~\cite{qin2020ultra} (ResNet-18~\cite{he_deep_2016} as the backbone) and PointPillars~\cite{lang2019pointpillars} for 3D object detection.

We assume all CAVs share their states, actions, and environment perception with others under the agreement constructed using Shapley value defined in Definition~\ref{def:Shapley_value_new}. We use a neural network to learn the characteristic function $v(s, a, \mathcal{C})$ and then use it to calculate the Shapley value as the reallocated total reward. Then we update the local policy by the Shapley value as shown in Alg.~\ref{alg:actor-critic}.
We use recent advances in safe RL: model predictive shielding (MPS) that has formal safety guarantee~\cite{li2020robust, zhang2019mamps} to maintain a safe learning process. 


\subsection{Comparison with Baselines}
We evaluate Algorithm~\ref{alg:actor-critic} against the state-of-the-art methods using 30 CAVs to show how Shapley value-based reward reallocation can be used to learn a cooperative policy to improve the system's total reward. We follow the open-source implementations from MADDPG~\cite{lowe2017multi}, M3DDPG~\cite{li2019robust}, and COMA~\cite{foerster2018counterfactual}. For all algorithms, agents are trained for 100k episodes with 5 random seeds and a maximum of 40 steps for each episode. We compare the mean episode system reward (system's total reward averaged over every 1000 episodes) as shown in Fig.~\ref{fig:baseline}. We observe that the Shapley value-based reward reallocation outperforms the state-of-the-art methods. This is because our method encourages more cooperation among vehicles. Vehicles can change lanes cooperatively to get a larger average velocity and comfort for the entire system. We also observe that our method converges slower because it relies on the estimation of the Shapley value which is not accurate at the beginning. The M3DDPG is a bit conservative because it considers the robustness to the worst-case scenario. 

\begin{figure}[h]
  \centering
  \includegraphics[width=3.4in]{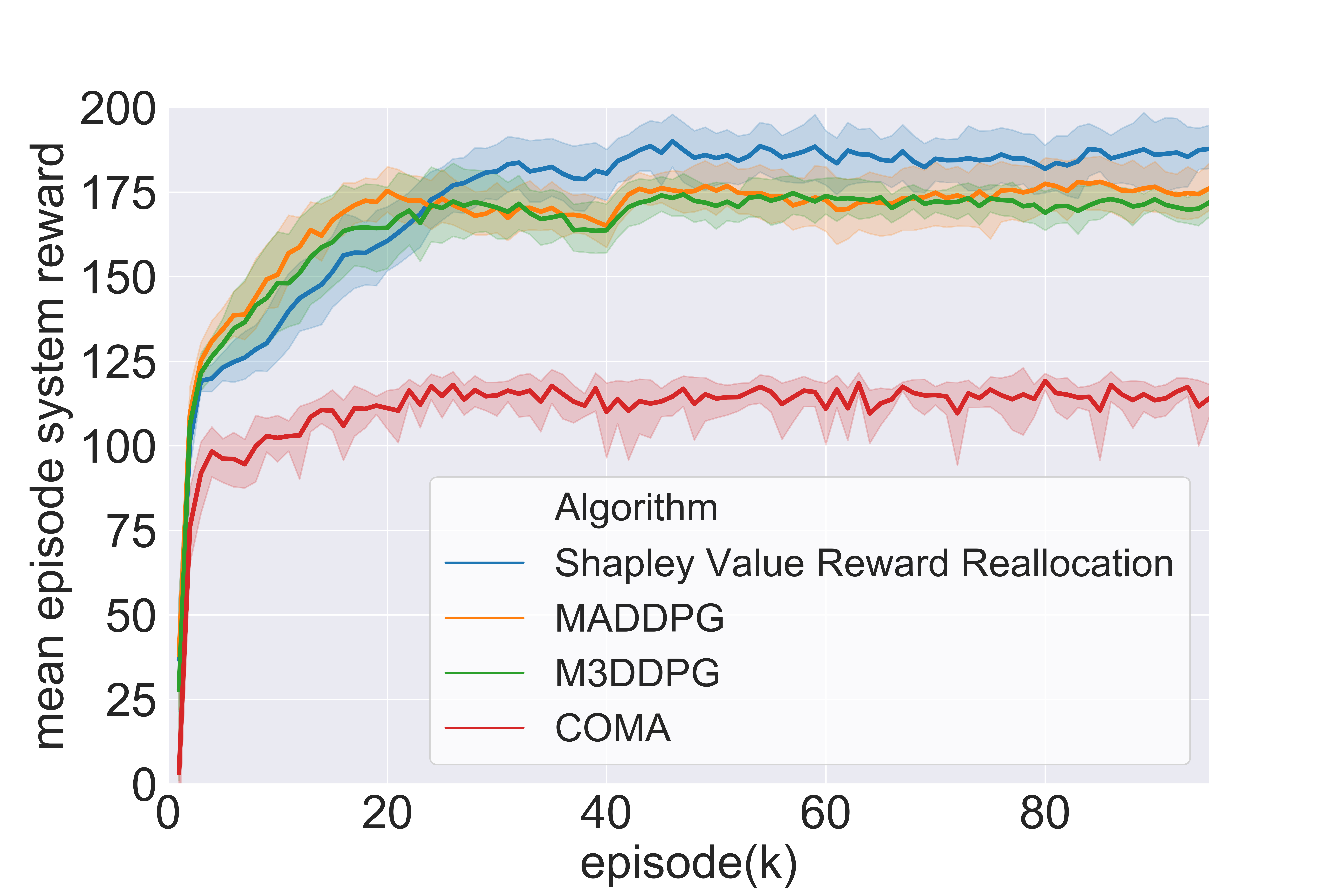}
  \vspace{-20pt}
  \caption{Comparison between our Shapley value reward reallocation method with baselines. Our method gets larger mean episode system reward.}\label{fig:baseline}
  \vspace{-10pt}
\end{figure}

\subsection{Mixed Traffic}


In this section, we show that our algorithm is also effective for CAVs in mixed traffic scenarios where the traffic includes both CAVs and human-driven vehicles. In this set of experiments, the total number of vehicles is 30. We change the CAV ratio (the total CAV number divided by the total number of all vehicles) from 0 to 1 as listed in Table~\ref{tbl:mixed_setting}. We use our Alg.~\ref{alg:actor-critic} and COMA for CAV's policy learning in two sets of experiments. We use CARLA's built-in autopilot mode to simulate human-driven vehicles~\cite{Dosovitskiy17}. As human-driven vehicle cannot share their states and actions, the CAVs can only form coalitions to estimate the centralized critic using $Q_{\pi^\mathcal{C}}(s^\mathcal{C}, a^\mathcal{C})$ that only has information from agents within coalition $\mathcal{C}$. The joint policy of coalition $\mathcal{C}$ is $\pi^\mathcal{C}(a^\mathcal{C} | s^\mathcal{C})$. In the implementation, we simply set the missing input of the neural network to be zero.

\begin{table}[h]
  \centering
  \caption{The system efficiency comparison under the mixed traffic.}
  \begin{tabular}{c|cc|cc}
  \toprule
  CAV & CAV  & human-driven & average  & average  \\
  ratio & number & vehicle number & velocity (mph) & comfort\\
  \midrule
  Algorithm & \multicolumn{4}{c}{Shapley Value Reward Reallocation} \\
  \midrule
  0 & 0 & 30 & 60.06 & 2.61 \\
  0.17 & 5 & 25 & 61.75 & 2.63 \\
  0.33 & 10 & 20 & 64.71 & 2.68 \\
  0.5 & 15 & 15 & 65.12 & 2.71 \\
  0.67 & 20 & 10 & 65.19 & 2.74 \\
  0.83 & 25 & 5 & 65.49 & 2.76 \\
  1 & 30 & 0 & 66.14 & 2.81 \\
  \midrule
  Algorithm & \multicolumn{4}{c}{COMA} \\
  \midrule
  0 & 0 & 30 & 60.06 & 2.61 \\
  0.17 & 5 & 25 & 61.19 & 2.62 \\
  0.33 & 10 & 20 & 63.15 & 2.65 \\
  0.5 & 15 & 15 & 63.44 & 2.68 \\
  0.67 & 20 & 10 & 63.48 & 2.69 \\
  0.83 & 25 & 5 & 63.69 & 2.71 \\
  1 & 30 & 0 & 64.12 & 2.75 \\
  \bottomrule
  \end{tabular}
  \label{tbl:mixed_setting}
\end{table}


We compare the average velocity and comfort for all vehicles under different CAV ratios. The velocity and comfort are averaged over all the 40,000 timesteps used in the simulation. The result in Table~\ref{tbl:mixed_setting} 
shows the average velocity and comfort of the entire mixed traffic. From the result using Shapley value reward reallocation, we can see the average velocity and comfort increase when the CAV ratio gets higher. Comparing the pure CAVs' case and the pure human-driven vehicles' case, the average velocity improves 10\% and the average comfort improves 8\%. The results also give us insights that the penetration of the CAVs can improve traffic in the future. Comparing with COMA, our algorithm gets a higher system's total reward in terms of both average velocity and comfort.

\section{conclusion}
In this work, we propose a Shapley value-based method to reallocate the system's total reward to each agent, to motivate cooperation among agents, for multi-agent systems such as connected autonomous vehicles (CAVs). We prove that the proposed Shapley value-based reward reallocation locates in the core of the convex game. Hence, the reward reallocation mechanism is stable and efficient, and each individual agent should stay in the cooperating coalition to receive more rewards. We design a cooperative policy learning algorithm which is centralized training and distributed execution. In experiments, we show the improvement of the system's total reward for CAV systems using the proposed algorithm. We also validate the effectiveness of our method in a mixed traffic scenario.
In the future, we will leverage the results to other networked CPS to better understand the benefits of communication-based cooperation.

\bibliographystyle{IEEEtran}
{ 
\bibliography{shapley}
}

\end{document}